%% file: CA-CNN_B.tex
\documentclass[11pt]{article}
\usepackage{mathpazo}
\usepackage[latin9]{inputenc}
\usepackage{verbatim}
\usepackage{calc}
\usepackage{mathrsfs}
\usepackage{mathtools}
\usepackage{url}
\usepackage{algorithm2e}
\usepackage{amsmath}
\usepackage{amssymb}
\usepackage{graphicx}

\makeatletter

\newcommand{\lyxmathsym}[1]{\ifmmode\begingroup\def\b@ld{bold}
  \text{\ifx\math@version\b@ld\bfseries\fi#1}\endgroup\else#1\fi}

\providecommand{\tabularnewline}{\\}

\usepackage{graphics}\usepackage{epsfig}

\newtheorem{theorem}{Theorem}
\newtheorem{thm}{Theorem}[section]

\newtheorem{lemma}[thm]{Lemma}
\newtheorem{corollary}[thm]{Corollary}

\newenvironment{proof}{\sl \noindent Proof: \rm}{$\Box$}
\newcommand{\rank}{{\rm rank}}
\newcommand{\Ker}{{\rm ker}}
\newcommand{\ignore}[1]{}

\title{Communication-Optimal Convolutional Neural Nets}

\author{
   James          Demmel\thanks{Computer Science Div.\ and Mathematics
                  Dept., Univ.\ of California, Berkeley, CA 94720
                  ({\tt demmel@berkeley.edu}).},\;
   Grace          Dinh\thanks{Computer Science Div.,\
                  Univ.\ of California, Berkeley, CA 94720
                  ({\tt dinh@berkeley.edu}).},\;
}
\date{\today}

\makeatother

\begin{document}
\maketitle

\input{abstract.tex}

\input{Introduction_B.tex} \input{CNN_def_B.tex}

\input{Matmul.tex} \input{LowerBound_B.tex} \input{UpperBound_B.tex}\input{Pooling.tex}

\input{Conclusions.tex}

\bibliographystyle{alpha}
\bibliography{biblio,grace-biblio}

\newpage{}

\input{UpperBound_B_manual.tex}

\end{document}

%% file: abstract.tex
\begin{abstract}
Efficiently executing convolutional neural nets (CNNs) is important
in many machine-learning tasks. Since the cost of moving a word of
data, either between levels of a memory hierarchy or between processors
over a network, is much higher than the cost of an arithmetic operation,
minimizing data movement is critical to performance optimization.
In this paper, we present both new lower bounds on data movement needed
for both convolutional and pooling layers of CNNs, and optimal sequential
algorithms that attain these lower bounds. In most common cases, our
optimal algorithms can attain significantly more data reuse than matrix
multiplication.
\end{abstract}

%% file: Introduction_B.tex
\section{Introduction}

\label{Sec_Intro} Convolutional neural networks are a bottleneck
in many machine learning applications, and as such must be efficiently
implemented on modern architectures. To do so, it is important to
understand where most of the time (and energy) goes when executing
a program on a current architecture. There are two costs to consider:
arithmetic and communication, i.e. moving data, either between levels
of a memory hierarchy or between processors over a network. The cost
to move one word of data can be orders of magnitude larger than the
cost to perform one arithmetic operation, and this difference in cost
is growing over time, following technological trends \cite{GSP05,FM11}.
Avoiding communication has long motivated optimization efforts (many
of which have in fact managed to attain known communication lower
bounds) in numerical linear algebra, resulting in tuned libraries
(e.g. the BLAS and LAPACK) that attain a high fraction of machine
peak. We seek to extend this optimization approach to CNNs.

In this paper, we consider one phase of CNNs, which can be written
most simply as seven nested loops combining a four-dimensional $Image$
array and a four-dimensional $Filter$ array to compute a four-dimensional
$Out$ array, which may be most simply stated as follows\footnote{We ignore boundary conditions in this paper and focus solely on asymptotic
optimization.}:
\begin{align*}
 & {\rm for}\,\{b,c,k,w,h,r,s\}=0:\{B,C,K,W,H,R,S\}-1\\
 & \ \ \ \ \;\;Out(k,h,w,b)+=Image(r+\sigma_{w}w,s+\sigma_{h}h,c,b)\times Filter(k,r,s,c)
\end{align*}
We consider all possible ways to reorganize this computation, performing
the same arithmetic operations in an arbitrary order, and ask which
order minimizes communication. We use a simple sequential architectural
model where there is one memory large enough to hold all the input
and output data, and a smaller cache of size $M$ where data needs
to reside to perform arithmetic; we wish to minimize the number of
words moved between the large memory and cache (we also show how to
generalize this simple architectural model to more complicated and
realistic ones).

Our first contribution is to prove new communication lower bounds
that hold for all possible loop bounds on the seven nested loops,
all strides, and all cache sizes $M$. Our second contribution is
to show how to attain these lower bounds in all cases, using appropriate
loop reorganizations and tilings.

Both the lower bounds and optimal tilings are more complicated than,
say, those for matrix multiplication because the seven loop bounds
and two stride values lead to many more possible situations than the
three loop bounds of matrix multiplication (which we will briefly
review for contrast). Not all loop bounds and strides may be used
in practice, but we consider them all for completeness.

Fortunately, the lower bound can be tersely written as the maximum
of five simple algebraic expressions in the seven loop bounds and
cache size $M$: 
\[
\max(BKWH,\sigma_{W}\sigma_{H}BCWH,CKRS,BCKWHRS/M,BCKWH(RS\sigma_{W}\sigma_{H}/M)^{1/2})
\]
(see Theorem~\ref{thm_LB} in section~\ref{sec_LowerBound}). Any
one of these five expressions may be much larger than the others,
depending on the loop bounds, the stride values, and $M$. Notice
that in many common cases (small filter, arrays too big to fit entirely
in cache), the fifth term in the above expression is the maximum.
If we were only able to achieve as much data reuse as in conventional
$O(n^{3})$ matrix multiplication (or the many other dense linear
algebra operations for which matrix multiply is a bottleneck), this
term would be equal to $KCHWBRS/M^{1/2}$. In contrast, we see our
bound is smaller (better) than this by a factor $\min(M^{1/2},(RS/\sigma_{h}\sigma_{w})^{1/2})$.
The proof (which can be skipped on a first reading) uses techniques
from functional analysis, group theory and lattice theory.

The optimal tilings that attain these lower bounds lead to a large
number of cases. We determine these cases by formulating the problem
of finding an optimal tiling as an optimization problem, where we
want to choose loop tile sizes that both fit in fast memory and maximize
the amount of work that can be done. By taking logarithms, this becomes
a linear program over the tiling parameters, with constraints restricting
the tiling parameters to feasible tilings that fit inside memory.
We show in section~\ref{sec_UpperBound} that for all possible values
of the loop bounds and strides, there is a feasible solution of this
linear program that attains the corresponding lower bound. A practical
implementation of our result might be done either by formulating and
solving the linear program, or by using the cases precomputed using
the method described in Section~\ref{sec_UpperBound}. For certain
sets of parameters found in real-world neural networks such as AlexNet,
our algorithms can produce an integer-factor reduction in the communication
cost over a matrix multiply-based approach for sufficiently small
(L1-L2 cache) values of $M$.

The remainder of this paper is organized as follows. Section~\ref{sec_CNN}
describes the seven-nested loop version of a CNN in more detail, and
the simplifications we make for the purpose of our analysis. 
Section~\ref{sec_Matmul} briefly reviews lower bounds and optimal
algorithms for matrix multiplication, to set the stage for our more
complicated analysis of CNNs. Section~\ref{sec_LowerBound} presents
our new lower bounds, and Section~\ref{sec_UpperBound} presents
the matching upper bounds, i.e. optimal algorithms. We extend our
analysis to pooling in Section \eqref{sec_pooling}.

%% file: CNN_def_B.tex
\section{CNN Model}

\label{sec_CNN}

As stated in the Introduction, we consider the following CNN computation:
\begin{align}
 & {\rm for}\,\{b,c,k,w,h,r,s\}=0:\{B,C,K,W,H,R,S\}-1\nonumber \\
 & \ \ \ \ \;\;Out(k,h,w,b)+=Image(r+\sigma_{w}w,s+\sigma_{h}h,c,b)\times Filter(k,r,s,c)\label{eqn_CNN}
\end{align}
where $Image$ has dimensions $(\sigma_{w}W+R)\times(\sigma_{h}H+S)\times C\times B$,
$Out$ has dimensions $K\times H\times W\times B$, and $Filter$
has dimensions $K\times R\times S\times C$. $B$ is the number of
images, $C$ is the number of input channels, $K$ is the number of
output channels, $W$ and $H$ are the width and height of the output
image, $R$ and $S$ are the sizes of one convolution, $\sigma_{w}$
is the stride size in the $w$ dimension, and $\sigma_{h}$ is the
stride size in the $h$ dimension. We assume that the filter size
is smaller than the input image size, i.e. $R\leq\sigma_{w}W$ and
$S\leq\sigma_{h}H$, and typically, they are much smaller, though
our analysis does not require this. We also assume that $\sigma_{w}\le R$
and $\sigma_{h}\le S$ (i.e. all elements of $Image$ are used in
the computation); we can reduce any problem onto one where this assumption
holds using no more communication than is necessary to read the useful
elements in \emph{Image}. Thus $Filter$ has total size $KCRS$, $Out$
has total size $KHWB$, and $Image$ has total size $C(\sigma_{h}H+S)(\sigma_{w}W+R)B$,
which is usually close to, and at most four times, $CHWB\sigma_{h}\sigma_{w}$.
These three array sizes will appear in our communication bounds. We
will use the expression $CHWB\sigma_{h}\sigma_{w}$ to simplify our
algebra later, since we are only interested in the asymptotics.

%% file: Matmul.tex
\section{Review of Matrix Multiplication}

\label{sec_Matmul}

In this section we review the well-known case of matrix multiplication
(matmul), both to compare to our analogous but more complicated result
for CNNs, and because the lower bound for matmul is used in the CNN
lower bound proof. We consider only ``classical'' matmul, i.e. the
algorithm that does $mnk$ multiplies and adds to multiply an $m$-by-$n$
matrix times an $n$-by-$k$ matrix (we discuss Strassen-like algorithms
briefly below). To keep it simple, we consider $n$-by-$n$ times
$n$-by-$k$ matmul $C=A*B$, where the matrix $A$ and $B$ originally
reside in main memory, the result $C$ is stored in main memory at
the end of execution, $1\leq k\leq n$, and the cache has size $M$.
In this case, the attainable lower bound on the number $W_{MM}$ of
words moved between main memory and cache is 
\begin{equation}
W_{MM}=\Omega(\max(n^{2},n^{2}k/M^{1/2}))
\end{equation}
The $n^{2}$ term arises because it is obviously necessary to read
the input matrices from main memory to cache at least once, and write
the output matrix to main memory at least once. The $n^{2}k/M^{1/2}$
term is the interesting one, dominating $n^{2}$ for large enough
$k$, and a decreasing function of $M$. In the square $n=k$ case,
it is attained by the well-known tiling approach, i.e. breaking matrices
$A$, $B$ and $C$ into square submatrices of dimension $(M/3)^{1/2}$,
so one submatrix of $A$, $B$ and $C$ can all fit in cache simultaneously.
The tiled algorithm then loops over tiles, multiplying two tiles of
$A$ and $B$ and updating one tile of $C$. As $k$ decreases, the
same tiling approach works until $k=(M/3)^{1/2}$, at which point
one tile just fits in the $n$-by-$k$ matrices $B$ and $C$, and
the two terms $n^{2}$ and $n^{2}k/M^{1/2}$ in $W_{MM}$ become equal
(to within a modest constant factor). As $k$ decreases further, the
$W_{MM}$ remains equal to $n^{2}$, and this is attained by continuing
to use an $(M/3)^{1/2}$-by-$(M/3)^{1/2}$ tile for $A$, but $(M/3)^{1/2}$-by-$k$
tiles for $B$ and $C$. We will see an analogous, but more complicated
transitioning of optimal tilings for CNNs.

The lower bound $n^{2}k/M^{1/2}$ was first derived for sequential
classical matmul in \cite{hongkung}, generalized to parallel implementations
in \cite{ITT04}, and to classical linear algebra more generally in
\cite{BallardDemmelHoltzSchwartz11}. We will apply a further generalization
of these bounds \cite{CDKSY13a,CDKSY15,KnightPhD} that apply to more
general nested loops accessing arrays to CNNs in section~\ref{sec_LowerBound}.
The case of Strassen-like matmul was addressed in \cite{CA_Strassen_JACM,JacobScottPhD}. 

%% file: LowerBound_B.tex
\section{Communication Lower Bounds}

\label{sec_LowerBound}

In this section, we state and prove our communication lower bound
for the CNN in (\ref{eqn_CNN}). We first state the bound for the
following basic memory model, and then show how to generalize it to
other models, following \cite{BallardDemmelHoltzSchwartz11}. We assume
the input data initially resides in a large main memory, and that
at the end of the computation the answer also resides in the main
memory. The main memory is attached to a cache of smaller size $M$,
to which data can be loaded, and from which data can be stored back
to main memory. Arithmetic can only be performed on data in cache,
with the operands and result of an operation needing to fit in cache.
Our goal is to find a lower bound the number of loads and stores needed
to complete the algorithm.

In the simplest case, when $M$ is large enough to hold all the inputs
and outputs, an optimal algorithm would move all the inputs from main
memory to cache, perform the algorithm without any more loads or stores,
and store the answer back to main memory at the end. This would attain
the trivial lower bound on the number of loads and stores, equal to
the size of all the inputs plus the size of all the outputs. The interesting
case is when $M$ is not large enough to do this.

In this case, following the approach of \cite{CDKSY13}, we will proceed
as follows: consider the algorithm as a sequence of instructions,
including \emph{load instructions} that transfer data from slow to
fast memory, \emph{store instructions} that transfer data from fast
to slow memory, and $F$ arithmetic operations. Break the sequence
into $R$ rounds of consecutive instructions, with each round containing
exactly $M$ load and store operations (with the possible exception
of the last); this provides a $2M$ upper bound on the amount of data
that can be used as input by operations in a single bound - the $M$
words already available at the beginning, and the at most $M$ words
loaded into fast memory at the beginning of a round. If we can show
that at most $G$ operations can be computed with $2M$ inputs and
outputs (i.e. in one round), then the number of rounds must be
\[
R\ge\left\lfloor F/G\right\rfloor 
\]
and the number of words transferred in the execution must be 
\begin{equation}
M\left\lfloor F/G\right\rfloor \ .\label{eqn_LB2}
\end{equation}

This approach can be used for more general memory models. For example,
we can bound the communication between two consecutive levels of a
multilevel memory hierarchy by treating the smaller and faster of
the two levels as the ``fast memory'' and everything slower and
larger than it as ``slow memory''. For distributed memory parallel
computations, following the approach of \cite{ITT04,BallardDemmelHoltzSchwartz11,CDKSY13,KnightPhD},
we can bound the memory traffic in and out of any node by treating
the memory on that node as the ``fast memory'', and the memory on
all the other processors as \char`\"{}slow memory\char`\"{}.

\subsection{Communication Lower Bounds for CNNs}

The main result of this section is \begin{theorem}\label{thm_LB}
Any execution order of (\ref{eqn_CNN}) moves $W_{CNN}$ words between
main memory and a cache of size $M$, where 
\begin{equation}
W_{CNN}=\Omega(\max(BKWH,\sigma_{W}\sigma_{H}BCWH,CKRS,BCKWHRS/M,BCKWH(RS\sigma_{W}\sigma_{H}/M)^{1/2}))\label{eqn_LB1}
\end{equation}
\end{theorem}

To provide some intuition for this, note that the first three terms
in the $\max()$ are the sizes of the output and inputs $Out$, $Image$
and $Filter$ respectively. The fourth term uses the results in \cite{CDKSY13a,CDKSY15},
which apply to general loop nests accessing arrays, though we will
see that more work is required to apply these results concretely.
This lower bound has the same power of $M$ in the denominator as
the direct n-body problem. The fifth term is new, and is larger than
the fourth term if and only if $RS<M\sigma_{W}\sigma_{H}$, which
is a common case.

\subsection{Proof of the lower bound $BCKWHRS/M$}

\label{subsec_LB2}

We will apply the general communication lower bounds for perfectly
nested loops accessing arrays, whose subscripts can be arbitrary affine
functions of the loops indices, that were developed in \cite{CDKSY13a,CDKSY15,KnightPhD}.
Without going deeply into the significant theory developed in these
papers, we sketch the approach, and the additional information we
need to apply it. Each loop iteration may be identified with a 7-tuple
of integers $(b,c,k,w,h,r,s)$, and the data required be in fast memory
to execute it by 3 projections $\phi_{1}(b,c,k,w,h,r,s)=(b,k,w,h)$
(the subscripts of $Out$), $\phi_{2}(b,c,k,w,h,r,s)=(b,c,r+\sigma_{w}w,s+\sigma_{h}h)$
(the subscripts of $Image$), and $\phi_{3}(b,c,k,w,h,r,s)=(c,k,r,s)$
(the subscripts of $Filter$). So if $V$ is a set of 7-tuples of
integers, $\phi_{1}(V)$, $\phi_{2}(V)$ and $\phi_{3}(V)$ represent
the sets of entries of $Out$, $Image$ and $Filter$, respectively,
needed to execute $V$. We seek a bound $G\geq|V|$, subject to $\phi_{1}(V)$,
$\phi_{2}(V)$ and $\phi_{3}(V)$ all fitting in fast memory, i.e.
$|\phi_{1}(V)|\leq M$, $|\phi_{2}(V)|\leq M$ and $|\phi_{3}(V)|\leq M$
(again, ignoring constant factors). The discrete Hölder-Brascamp-Lieb
(HBL) inequalities developed in the above (and many previous) publications
tell us that there are nonnegative constants $s_{1}$, $s_{2}$ and
$s_{3}$ such that for all finite $V$,
\begin{equation}
|V|\leq\prod_{i=1}^{3}|\phi_{i}(V)|^{s_{i}}\label{eqn_HBL_1}
\end{equation}
which implies 
\begin{equation}
|V|\leq G=\prod_{i=1}^{3}M^{s_{i}}=M^{\sum_{i=1}^{3}s_{i}}\label{eqn_HBL_2}
\end{equation}
is the bound we seek. A set $(s_{1},s_{2},s_{3})$ satisfies (\ref{eqn_HBL_1})
for all $V$ if and only if they satisfy the linear inequalities 
\begin{equation}
\rank(H)\leq\sum_{i=1}^{3}s_{i}\cdot\rank(\phi_{i}(H))\label{eqn_HBL_3}
\end{equation}
for all subgroups $H$ of the abelian group $\mathbb{Z}^{7}$ of 7-tuples
of integers under addition, and where $\rank(H)$ is analogous to
the dimension of a vector space. Since there are an infinite number
of possible subgroups $H$, this looks like an infinite number of
inequalities, but in fact there are only finitely many, since each
$\rank(H)$ and $\rank(\phi_{i}(H))$ is an integer between 0 and
7. So to get the best (smallest) bound $G$, we want to minimize $\sum_{i=1}^{3}s_{i}$
subject to (\ref{eqn_HBL_3}), a linear program.

It will turn out that the minimal value of $\sum_{i=1}^{3}s_{i}$
is 2, leading to $G=M^{2}$, and using (\ref{eqn_LB2}) the desired
lower bound of 
\begin{equation}
M\lfloor F/G\rfloor=M\lfloor BCKWHRS/M^{2}\rfloor=O(BCKWHRS/M)
\end{equation}

The challenge is identifying a finite set of subgroups $H$ that generate
enough inequalities (\ref{eqn_HBL_3}) to get the correct solution
to the linear program. An algorithm for this is proposed in \cite{CDKSY15},
with a sketch of a simpler one based on \cite{Valdimarsson10}, which
is the approach we take. This requires us to generate the {\em lattice
of subgroups} generated by the kernels of $\phi_{1}$, $\phi_{2}$
and $\phi_{3}$. A lattice of subgroups (see \cite{Birkhoff} for
more background) is the set of all possible sums and intersections
that can be generated starting from a set of generators, i.e. subgroups
(we note that all subgroups discussed here are subgroups of ${\mathbb{Z}}^{7}$
and so abelian). Since the sum or intersection of two subgroups is
a subgroup, the lattice consists of subgroups of ${\mathbb{Z}}^{7}$.
We will always include $\{0\}$ in our lattices, since this does not
change other members of the lattice, and simplifies some expressions
below. We need some machinery to help describe the lattice we need
(or a superset) in a finite way.

Suppose $A=\{A_{1},...,A_{n}\}$ and $B=\{B_{1},...,B_{m}\}$ are
finite sets of subgroups of an abelian group. We will call them \textit{independent}
if 
\[
\sum_{i}A_{i}\cap\sum_{j}B_{j}=\{0\}
\]
Let ${\rm lattice}(A)$ denote the lattice generated by the subgroups
in $A$, and similarly for ${\rm lattice}(B)$; recall that we will
add $\{0\}$ to these lattices if they do not already contain it.
Then from the definition of a lattice it is easy to see that ${\rm lattice}(A)$
and ${\rm lattice}(B)$ are independent if $A$ and $B$ are independent,
in which case we have

\begin{lemma}\label{lemma_IndepLattice} Suppose $A$ and $B$ are
independent. Then 
\[
{\rm lattice}(A\cup B)={\rm lattice}(A)+{\rm lattice}(B)\equiv\{C+D:C\in{\rm lattice}(A),D\in{\rm lattice}(B)\}
\]
\end{lemma} \begin{proof} $A\cup B\subset{\rm lattice}(A)+{\rm lattice}(B)$
since both lattices include $\{0\}$. It suffices to show that if
$A_{1}+B_{1}$ and $A_{2}+B_{2}$ are both in ${\rm lattice}(A)+{\rm lattice}(B)\subseteq{\rm lattice}(A\cup B)$,
then so are their sum and intersection. The sum $(A_{1}+B_{1})+(A_{2}+B_{2})=(A_{1}+A_{2})+(B_{1}+B_{2})\in{\rm lattice}(A)+{\rm lattice}(B)$
follows from being abelian, and a lattice being closed under addition.
The intersection $(A_{1}+B_{1})\cap(A_{2}+B_{2})=(A_{1}\cap A_{2})+(B_{1}\cap B_{2})$
follows from independence, and a lattice being closed under intersection.
\end{proof}

The advantage of this is that if we have simple descriptions of ${\rm lattice}(A)$
and ${\rm lattice}(B)$ (say finite lists of each), then it is easy
to describe ${\rm lattice}(A\cup B)$.

More generally, let $C_{i}=\{C_{i,1},...,C_{i,n(i)}\}$ be a set of
$n(i)$ subgroups, for $i=1,...,m$. We call them \textit{independent}
if for all $i$ 
\[
\sum_{j}C_{i,j}\cap\sum_{k\neq i}\sum_{j}C_{k,j}=\{0\}
\]
This is a natural generalization of the case where each $C_{i,j}$
is a vector space. Then, as above the ${\rm lattice}(C_{i})$ are
independent, and 
\[
{\rm lattice}(\cup_{i}C_{i})=\sum_{i}{\rm lattice}(C_{i})
\]
and if we have finite lists of the members of each ${\rm lattice}(C_{i})$,
we can also list the members of ${\rm lattice}(\cup_{i}C_{i})$.

Now suppose we start with some sets of the form 
\begin{equation}
D_{1}=\sum_{i}C_{i,m(1,i)},\;\;\ldots\;\;,D_{k}=\sum_{i}C_{i,m(k,i)}\label{eqn_Ds}
\end{equation}
Each of the above sums may be over a subset of the possible values
of $i$. In other words, each $D_{j}$ is gotten by choosing at most
one member of each $C_{i}$, and adding them. Later we will choose
these $D_{j}$ to be the kernels of the projections $\phi_{i}$ in
the CNN. Then for any such $D_{1},\cdots,D_{k}$, we have 
\begin{equation}
{\rm lattice}(\cup_{i}D_{i})\subseteq{\rm lattice}(\cup_{i}C_{i})=\sum_{i}{\rm lattice}(C_{i})\label{eqn_lattice}
\end{equation}

Now we tie this to code for CNNs. The projections $\phi_{i}$ and
their kernels are 
\begin{eqnarray*}
\phi_{1}((b,c,k,w,h,r,s)) & =(b,k,w,h), & \Ker(\phi_{1})=\{(0,c,0,0,0,r,s)\}\\
\phi_{2}((b,c,k,w,h,r,s)) & =(b,c,r+\sigma_{w}w,s+\sigma_{h}h), & \Ker(\phi_{2})=\{(0,0,k,w,h,-\sigma_{w}w,-\sigma_{h}h)\}\\
\phi_{3}((b,c,k,w,h,r,s)) & =(c,k,r,s), & \Ker(\phi_{3})=\{(b,0,0,w,h,0,0)\}
\end{eqnarray*}
Our goal is a finite list of subgroups $H$ that is a superset of
${\rm lattice}(K)$, where 
\[
K=\{\Ker(\phi_{1}),\;\Ker(\phi_{2}),\;\Ker(\phi_{3})\}\;,
\]
and where for each such $H$ we can write down the inequality (\ref{eqn_HBL_3}).
Then by \cite{CDKSY15,Valdimarsson10}, solving the linear program
that minimizes $\sum_{i=1}^{3}s_{i}$ subject to these constraints
will give us our desired bound $G=M^{\sum_{i=1}^{3}s_{i}}$.

There are 5 groups of subscripts, $\{k\}$, $\{h,s,s+\sigma_{h}h\}$,
$\{w,r,r+\sigma_{w}w\}$, $\{c\}$, and $\{b\}$, that are independent
of one another. From these we define the following 5 sets of subgroups:
\begin{eqnarray*}
C_{1} & = & \{(0,0,k,0,0,0,0)\}=\{C_{1,1}\}\\
C_{2} & = & \{(0,0,0,0,h,0,0),(0,0,0,0,0,0,s),(0,0,0,0,h,0,-\sigma_{h}h)\}=\{C_{2,1},C_{2,2},C_{2,3}\}\\
C_{3} & = & \{(0,0,0,w,0,0,0),(0,0,0,0,0,r,0),(0,0,0,w,0,-\sigma_{w}w,0)\}=\{C_{3,1},C_{3,2},C_{3,3}\}\\
C_{4} & = & \{(0,c,0,0,0,0,0)\}=\{C_{4,1}\}\\
C_{5} & = & \{(b,0,0,0,0,0,0)\}=\{C_{5,1}\}
\end{eqnarray*}
Then we can write 
\begin{eqnarray*}
\Ker(\phi_{1}) & = & C_{2,2}+C_{3,2}+C_{4,1}\\
\Ker(\phi_{2}) & = & C_{1,1}+C_{2,3}+C_{3,3}\\
\Ker(\phi_{3}) & = & C_{2,1}+C_{3,1}+C_{5,1}
\end{eqnarray*}
which we identify with $D_{1}$, $D_{2}$ and $D_{3}$ in (\ref{eqn_Ds})
above. So by (\ref{eqn_lattice}), all we need are finite lists of
members of each ${\rm lattice}(C_{i})$ to write down a finite list
of subgroups $H$ containing ${\rm lattice}(K)$. Since each $C_{i}$
is small, it is easy to confirm the following facts: 
\begin{eqnarray*}
{\rm lattice}(C_{1}) & = & C_{1}\cup\{0\}\\
{\rm lattice}(C_{2}) & = & C_{2}\cup\{(0,0,0,0,h,0,s)\}\cup\{0\}\\
{\rm lattice}(C_{3}) & = & C_{3}\cup\{(0,0,0,w,0,r,0)\}\cup\{0\}\\
{\rm lattice}(C_{4}) & = & C_{4}\cup\{0\}\\
{\rm lattice}(C_{5}) & = & C_{5}\cup\{0\}
\end{eqnarray*}
where we have added $\{0\}$ to ${\rm lattice}(C_{1})$, ${\rm lattice}(C_{4})$
and ${\rm lattice}(C_{5})$. Since the cardinalities of these five
lattices are 2, 5, 5, 2 and 2, respectively, the number of possibly
different subgroups in $\sum_{i}{\rm lattice}(C_{i})$ in (\ref{eqn_lattice})
is at most $2\cdot5\cdot5\cdot2\cdot2=200$.

It turns out that we only need $1+4+4+1+1=11$ subgroups, or more
generally $\sum_{i}(|C_{i}|-1)$, not $\prod_{i}|C_{i}|$. This simplification
is a generalization of the ``Product Case'' in section 6.3 of \cite{CDKSY13a}.
The idea is that if $H\in\sum_{i}{\rm lattice}(C_{i})$, so that $H=\sum_{i}C_{i,j(i)}$
where $C_{i,j(i)}\in{\rm lattice}(C_{i})$, then by independence of
the $C_{i}$ we get $\rank(H)=\sum_{i}\rank(C_{i,j(i)})$, and we
also get $\rank(\phi_{k}(H))=\sum_{i}\rank(\phi_{k}(C_{i,j(i)}))$
by the construction of the $C_{i}$ from the $\phi_{k}$. Thus (\ref{eqn_HBL_3})
follows from adding all the inequalities 
\begin{equation}
\rank(C_{i,j(i)})\leq\sum_{k=1}^{3}s_{k}\cdot\rank(\phi_{k}(C_{i,j(i)}))\ .\label{eqn_HBL_4}
\end{equation}
There are only 11 such inequalities, because using $C_{i,j(i)}=\{0\}$
only yields the trivial inequality $0\leq0$.

The table below has one row for each $C_{i,j(i)}$, one column for
$\rank(C_{i,j(i)})$, 3 columns for each $\rank(\phi_{k}(C_{i,j(i)}))$,
and the rightmost column for the resulting inequality (\ref{eqn_HBL_4}).

\begin{center}
\begin{tabular}{|l|c|c|c|c|l|}
\hline 
$C_{i,j(i)}$  & $\rank(C_{i,j(i)})$  & $\rank(\phi_{1}(C_{i,j(i)}))$  & $\rank(\phi_{2}(C_{i,j(i)}))$  & $\rank(\phi_{3}(C_{i,j(i)}))$  & Inequality (\ref{eqn_HBL_4}) \tabularnewline
\hline 
\hline 
$C_{1,1}$  & 1  & 1  & 0  & 1  & $1\leq s_{1}+s_{3}$ \tabularnewline
$C_{2,1}$  & 1  & 1  & 1  & 0  & $1\leq s_{1}+s_{2}$ \tabularnewline
$C_{2,2}$  & 1  & 0  & 1  & 1  & $1\leq s_{2}+s_{3}$ \tabularnewline
\hline 
$C_{2,3}$  & 1  & 1  & 0  & 1  & $1\leq s_{1}+s_{3}$ \tabularnewline
$C_{2,4}$  & 2  & 1  & 1  & 1  & $2\leq s_{1}+s_{2}+s_{3}$ \tabularnewline
$C_{3,1}$  & 1  & 1  & 1  & 0  & $1\leq s_{1}+s_{2}$ \tabularnewline
\hline 
$C_{3,2}$  & 1  & 0  & 1  & 1  & $1\leq s_{2}+s_{3}$ \tabularnewline
$C_{3,3}$  & 1  & 1  & 0  & 1  & $1\leq s_{1}+s_{3}$ \tabularnewline
$C_{3,4}$  & 2  & 1  & 1  & 1  & $2\leq s_{1}+s_{2}+s_{3}$ \tabularnewline
\hline 
$C_{4,1}$  & 1  & 0  & 1  & 1  & $1\leq s_{2}+s_{3}$ \tabularnewline
$C_{5,1}$  & 1  & 1  & 1  & 0  & $1\leq s_{1}+s_{2}$ \tabularnewline
\hline 
\end{tabular}
\par\end{center}

Removing redundant inequalities, we get just the following four: 
\[
1\leq s_{1}+s_{2},\;1\leq s_{1}+s_{3},\;1\leq s_{2}+s_{3},\;2\leq s_{1}+s_{2}+s_{3}
\]
We see that minimizing $\sum_{i=1}^{3}s_{i}$ subject to these inequalities
yields the desired value of 2, say by choosing $s_{1}=s_{2}=s_{3}=2/3$.
The solution $(s_{1},s_{2},s_{3})$ is not unique, but their sum is.

\subsection{Proof of the lower bound $BCKWH(\frac{RS\sigma_{W}\sigma_{H}}{M})^{1/2}$}

\label{subsec_LB3}

\cite{RuscianoDemmel16} shows how to attain the communication lower
bound (\ref{eqn_LB2}) for \emph{any} algorithm expressible as perfectly
nested loops accessing arrays whose subscripts are all affine functions
of the loop indices, including CNNs. It does this by showing how to
construct an optimal tiling in all such cases (tilings are explained
in more detail in Section \ref{sec_UpperBound}). But this is not
the end of the story, because \cite{RuscianoDemmel16} assumes the
loop bounds are big enough to fit an entire tile. For CNNs, however,
this is often not the case: if the size $RS$ of individual convolution
is sufficiently small, the optimal tile size given by \cite{RuscianoDemmel16}
may have block sizes for $r$ and $s$ bigger than the array bounds
$R$ and $S$. In this case, a tighter lower bound holds.

Using the notation introduced above, we still want to bound the number
of lattice points $|V|$ in any set $V$ of 7-tuples $(b,c,k,w,h,r,s)$
of integers, given the bounds $|\phi_{i}(V)|\leq M$ for $i\in\{1,2,3\}$.
However, we want a bound that is tighter than $M^{2}$ when $RS$
is small.

We will begin by rewriting the loop indices as $r=\sigma_{w}r'+r''$
and $s=\sigma_{h}s'+s''$, where $r''\in[0,\sigma_{w}-1]$ and $s''\in[0,\sigma_{h}-1]$.
Replace the loop over $r$ with loops over $r'$ (from $0$ to $R/\sigma_{w}-1)$
and $r''$ (from $0$ to $\sigma_{w}-1$), and replace the loop over
$s$ similarly. Since each $r,s$ maps uniquely onto a single $r',r'',s',s''$,
we can lift both $Image$ and $Filter$ onto a higher dimension.

There is a one-to-one correspondence between every point in the original
seven-dimensional lattice and every point in the lifted nine-dimensional
lattice. Therefore, it suffices to bound the number of lattice points
$\vert V\vert$ of any set $V$ of $9$-tuples $(b,c,k,w,h,r',r'',s',s'')$
such that $|\phi'_{i}(V)|\leq M$, with $\phi'_{i}$ defined as follows:

\begin{eqnarray*}
\phi'_{1}(b,c,k,w,h,r',r'',s',s'') & = & (b,k,w,h)\\
\phi'_{2}(b,c,k,w,h,r',r'',s',s'') & = & (b,c,r',r'',w,s',s'',h)\\
\phi'_{3}(b,c,k,w,h,r',r'',s',s'') & = & (c,k,r',r'',s',s'')
\end{eqnarray*}

\begin{lemma}\label{lemma_LB_RSsmall} Let $V$ be any set of 9-tuples
of integers $(b,c,k,w,h,r',r'',s',s'')$ where $|\phi'_{i}(V)|\leq M$
for $i\in\{1,2,3\}$, and $1\leq r'\leq R/\sigma_{W}$ and $1\leq s\leq S/\sigma_{H}$.
Then 
\[
|V|\leq\frac{(RS)^{1/2}M^{3/2}}{(\sigma_{W}\sigma_{H})^{1/2}}\ .
\]
 \end{lemma}

This bound is obviously tighter than $M^{2}$ precisely when $RS<M\sigma_{W}\sigma_{H}$.
Plugging $G=(RS)^{1/2}M^{3/2}(\sigma_{W}\sigma_{H})^{-1/2}$ into
(\ref{eqn_LB2}) immediately yields:

\begin{corollary}\label{cor_LB_RSsmall} The number of reads and
writes to execute a CNN with a fast memory of size $M$ is at least
$BCKWH(RS\sigma_{W}\sigma_{H}/M)^{1/2}$. \end{corollary}

\noindent \textsl{Proof of Lemma~\ref{lemma_LB_RSsmall}:} Let $V(r',s')$
denote the restriction of $V$ to a given value of $r',s'$, so that
$V=\cup_{r',s'}V(r',s')$ is a disjoint union of sets and $|V|=\sum_{r',s'}|V(r',s')|$.
Note that $\phi'_{3}(V)=\cup_{r',s'}\phi'_{3}(V(r',s'))$ is also
a disjoint union of sets, so 
\begin{equation}
|\phi'_{3}(V)|=\sum_{r',s'}|\phi'_{3}(V(r',s'))|\leq M\label{eqn_LBVrs}
\end{equation}
Also $|\phi'_{1}(V(r',s'))|\leq|\phi'_{1}(V)|\leq M$ and $|\phi'_{2}(V(r',s'))|\leq|\phi'_{2}(V)|\leq M$.
We want to bound $|V(r',s')|$ in terms of these bounds on $|\phi'_{1}(V(r',s'))|$,
$|\phi'_{2}(V(r',s'))|$, and $|\phi'_{3}(V(r',s'))|$.

This is another application of the HBL inequalities discussed in section
\ref{subsec_LB2}. Since each of the loop indices appear in exactly
two of the three $\phi'_{i}$, this is a special case of a tensor
contraction, for which the optimal exponents are $s_{1}=s_{2}=s_{3}=1/2$
as shown in Section 6.3 of \cite{CDKSY13a}. This yields:
\begin{eqnarray*}
|V| & = & \sum_{r',s'}|V(r',s')|\\
 & \leq & \sum_{r',s'}|\phi'_{1}(V(r',s'))|^{1/2}\cdot|\phi'_{2}(V(r',s'))|^{1/2}\cdot|\phi'_{3}(V(r',s'))|^{1/2}\\
 & \leq & \sum_{r',s'}M^{1/2}\cdot M^{1/2}\cdot|\phi'_{3}(V(r',s'))|^{1/2}\\
 & = & M\cdot\sum_{r',s'}|\phi'_{3}(V(r',s'))|^{1/2}
\end{eqnarray*}
We want to maximize this subject to (\ref{eqn_LBVrs}). Since we are
summing $RS/(\sigma_{w}\sigma_{h})$ terms, a simple application of
Lagrange multipliers tells us that this maximum is attained when all
$|\phi'_{3}(V(r',s'))|=M/(RS/(\sigma_{w}\sigma_{h}))$, yielding the
desired 
\[
|V|\leq(\frac{RS}{\sigma_{w}\sigma_{h}})^{1/2}M^{3/2}\;\;.
\]
$\Box$ 

%% file: UpperBound_B.tex
\section{Communication Optimal Algorithms}

\label{sec_UpperBound}

In this section we show that the lower bounds in Theorem~\ref{thm_LB}
are always attainable by an appropriate tiling, analogous to the one
for matrix multiplication in section~\ref{sec_Matmul}. We will prove
this by constructing a tiling for any possible set of array bounds,
and verifying that the tiling attains one of lower bounds from \eqref{eqn_LB1}.
The tile sizes we construct may also be useful starting points for
optimization in practice, although the exact tile sizes may not necessarily
give the best performance (due to constant factors omitted from our
analysis).

Following the approach of \cite{RuscianoDemmel16}, we will achieve
this tiling by blocking each variable into contiguous blocks, except
for $r$ and $s$, which we will rewrite as $r=\sigma_{w}r'+r''$
and $s=\sigma_{h}s'+s''$ (with $r''\in[0,\sigma_{w}-1]$ and $s''\in[0,\sigma_{h}-1]$)
respectively. That is, we will rewrite our loop as follows, where
we use ${\rm for}\,i=\alpha:\beta:\gamma$ to denote iterating from
$\alpha$ to $\gamma$ with a step size of $\beta$:
\begin{align*}
 & {\rm for}\,\{b,c,k,w,h\}_{1}=0:b_{\{b,c,k,w,h\}}:\{B,C,K,W,H\}-b_{\{b,c,k,w,h\}},\\
 & \;\;{\rm for}\,r'_{1}=0:b_{r'}:R/\sigma_{w}-b_{r'},\;{\rm for}\,r''_{1}=0:b_{r''}:\sigma_{w}-b_{r''},\\
 & \ \ \ \ {\rm for}\,s'_{1}=0:b_{s'}:S/\sigma_{h}-b_{s'},\;{\rm for}\,s''_{1}=0:b_{s''}:\sigma_{h}-b_{s''},\\
 & \ \ \ \ \ \ {\rm for}\,\{b,c,k,w,h\}_{2}=0:b_{\{b,c,k,w,h\}}-1\\
 & \;\;\ \ \ \ \ \ {\rm for}\,r'_{2}=0:b_{r'}-1,\;{\rm for}\,r''_{2}=0:b_{r''}-1,\\
 & \ \ \ \ \ \ \ \ \ \ {\rm for}\,s'_{2}=0:b_{s'}-1,\;{\rm for}\,s''_{2}=0:b_{s''}-1,\\
 & \;\;\;\;\;\;\ \ \ \ \ \ \{b,c,k,w,h,r',r'',s',s''\}=\{b,c,k,w,h,r',r'',s',s''\}_{1}+\{b,c,k,w,h,r',r'',s',s''\}_{2}\\
 & \;\;\;\;\;\;\ \ \ \ \ \ Out(k,h,w,b)+=Image(r''+\sigma_{w}(r'+w),\ s''+\sigma_{h}(s'+h),c,b)\\
 & \;\;\;\;\;\;\ \ \ \ \ \ \phantom{Out(k,h,w,b)+=}\times Filter(k,\sigma_{w}r'+r'',\ \sigma_{h}s'+s'',s,c)
\end{align*}

To minimize the communication cost, it suffices to maximize the size
of each block (that is, $b_{b}b_{c}b_{k}b_{w}b_{h}b_{r'}b_{r''}b_{s'}b_{s''}$)
subject to the following constraints:
\begin{enumerate}
\item Each block size must be positive:
\[
b_{\{b,c,k,w,h,r',r'',s',s''\}}\ge1
\]
\item The block size in each dimension is smaller than the loop bound for
that dimension. For the first five indices, we have:
\[
b_{\{b,c,k,w,h\}}\le\{B,C,K,W,H\}
\]
The loop bounds on $r''$ and $s''$ (given by their definitions)
give the following two constraints:
\begin{eqnarray*}
b_{r''} & \le & \sigma_{w}\\
b_{s''} & \le & \sigma_{h}
\end{eqnarray*}
To ensure that the blocks for $r'$ and $s'$ are of appropriate size,
recall that $r=\sigma_{w}r'+r''$ and $s=\sigma_{h}s'+s''$, which
gives
\begin{eqnarray*}
\sigma_{w}b_{r'}+b_{r''} & \le & R\\
\sigma_{h}b_{s'}+b_{s''} & \le & S
\end{eqnarray*}
Since $b_{r''}\le\sigma_{w}$ and $b_{s''}\le\sigma_{h}$, we can
safely omit those from the inequality (since their effect on left-hand
side is at most equivalent to adding $1$ to $b_{r'}$ and $b_{s'}$,
and we are only interested in asymptotics) to get
\begin{eqnarray*}
\sigma_{w}b_{r'} & \le & R\\
\sigma_{h}b_{s'} & \le & S
\end{eqnarray*}
\item The size of each block does not exceed the size of fast memory $M$.
This is straightforward for $Out$:
\[
b_{b}b_{k}b_{w}b_{h}\le M
\]
as well as $Filter$:
\[
b_{c}b_{k}b_{r'}b_{r''}b_{s'}b_{s''}\le M\ .
\]
For $Image$, notice that if a block for $r'$ is $[r'_{start},r'_{end}]$,
a block of $r''$ is $[r''_{start},r''_{end}]$ and a block for $w$
is $[w_{start},w_{end}]$, then the indices of $Image$ in the $w$-dimension
accessed will be of the form $i+\sigma_{w}j$, where $i\in[r''_{start},r''_{end}]$
and $j\in[w_{start}+r'_{start},w_{end}+r'_{end}]$. As a result, the
number of indices in the $w$-dimension accessed is $(b_{w}+b_{r'})b_{r''}$;
similarly for the $h$-dimension. Therefore, the total number of elements
accessed from $Image$ must be:
\[
b_{b}b_{c}(b_{w}+b_{r'})(b_{h}+b_{s'})b_{r''}b_{s''}\le M\ .
\]
As we will see shortly, it is convenient to recast our maximization
problem as a linear program by taking logs; for this to happen, we
only want products in the inequality. Multiplying out the left-hand
side of the above gives a sum of four terms; bounding each of them
by $M$ is sufficient for an asymptotic analysis. Therefore, we get:
\begin{eqnarray*}
b_{b}b_{c}b_{w}b_{h}b_{r''}b_{s''} & \le & M\\
b_{b}b_{c}b_{w}b_{s'}b_{r''}b_{s''} & \le & M\\
b_{b}b_{c}b_{r'}b_{h}b_{r''}b_{s''} & \le & M\\
b_{b}b_{c}b_{r'}b_{s'}b_{r''}b_{s''} & \le & M
\end{eqnarray*}
\end{enumerate}
Taking the log base $M$ of the objective and all the constraints,
we get the following linear program, with $l_{\{b,...,s''\}}=\log_{M}b_{\{b,...,s''\}}$:
\begin{eqnarray}
\max l_{b}+l_{c}+l_{k}+l_{w}+l_{h}+l_{r'}+l_{r''}+l_{s'}+l_{s''}\ s.t.\nonumber \\
l_{\{b,c,k,w,h,r',r'',s',s''\}} & \ge & 0\nonumber \\
l_{\{b,c,k,w,h\}} & \le & \{\log_{M}B,\log_{M}C,\log_{M}K,\log_{M}W,\log_{M}H\}\nonumber \\
l_{r''} & \le & \log_{M}\sigma_{w}\nonumber \\
l_{s''} & \le & \log_{M}\sigma_{h}\nonumber \\
\log_{M}\sigma_{w}+l_{r'} & \le & \log_{M}R\nonumber \\
\log_{M}\sigma_{h}+l_{s'} & \le & \log_{M}S\nonumber \\
l_{b}+l_{k}+l_{w}+l_{h} & \le & 1\nonumber \\
l_{c}+l_{k}+l_{r'}+l_{r''}+l_{s'}+l_{s''} & \le & 1\nonumber \\
l_{b}+l_{c}+l_{w}+l_{h}+l_{r''}+l_{s''} & \le & 1\nonumber \\
l_{b}+l_{c}+l_{w}+l_{s'}+l_{r''}+l_{s''} & \le & 1\nonumber \\
l_{b}+l_{c}+l_{r'}+l_{h}+l_{r''}+l_{s''} & \le & 1\nonumber \\
l_{b}+l_{c}+l_{r'}+l_{s'}+l_{r''}+l_{s''} & \le & 1\label{eq:mplp}
\end{eqnarray}

We will first determine a closed-form solution for this linear program;
that is, we will partition the space of possible input parameters
(i.e. the array bounds and strides) into convex polytopes, within
each of which the optimal tiling and communication cost it achieves
are described by a single linear function of the input parameters.
We will then use this closed-form solution to verify optimality of
the tiling (and attainability of the lower bound) by showing that
the communication cost of the optimal tiling is always equal to a
communication lower bound for every point in the parameter space.

Although this construction and verification can in theory be performed
by hand (see Appendix A for a hand analysis of the case where $\sigma_{w}=\sigma_{h}=1$),
the size of the result - the partition we find is a set of $200$
regions - makes it far more expedient to automate the analysis, which
will also allow us to more easily extend this approach to other problems.
An implementation of the algorithm, as well as a table of partitions,
optimal tilings, and optimal communication costs in each of these
partitions, may be found at \url{https://people.eecs.berkeley.edu/~dinh/papers/DD18/partitioning.nb}.

\subsection{Determining the optimal tiling}

\label{subsec:opttile}

Algorithms for determining a closed-form solution to the parameterized
linear programs have been studied extensively in the context of control
theory \cite{GN72,BBM03,STJ05,JBM07}. We used the geometric algorithm
from \cite{BBM03}, which we briefly describe in this section and
fully specify in Figure \ref{fig:mplp}.; see the original paper for
a proof of correctness.

For convenience, represent the LP \eqref{eq:mplp} as
\begin{eqnarray*}
\min c^{T}x & s.t\\
Gx & \le & w+F\theta
\end{eqnarray*}
 where $c=[-1,...,-1]^{T}$, $x=[l_{b},l_{c},...,l_{s''}]$, $G$
is the coefficient matrix for the left-hand side of the inequalities,
$\theta=[\log_{M}B,\log_{M}C,...,\log_{M}\sigma_{h}]^{T}$, and $F$
and $w$ are the coefficient matrix and vector, respectively, for
the right-hand side of the inequalities.

The intuition for the algorithm is as follows: start with a (possibly
open) polytope in parameter space; during the first iteration of the
algorithm, this is the set of all possible valid loop bounds and strides
(i.e. nonnegative parameters, filter fitting inside the input). Pick
a random point (not necessarily uniformly) inside that region, setting
the parameters to its coordinates. Solve the linear program at that
point using a method, such as simplex, that guarantees that the solution
produced will be a vertex of the polytope, and note which constraints
are made tight.

The number of tight constraints should be at least to the number of
variables in the linear program, which is nine in this case, since
our solution is a vertex of the polytope. If there are more than nine
constraints, there are two possibilities: either (a) when the constraints
are set to equality (to ensure tightness), there are redundant constraints,
and the number of non-redundant constraints is nine, or (b) the point
we selected lies on the border of two partitions in parameter space.
Since the borders of partitions are of lower dimension than the parameter
space itself (and since we selected our initial point randomly), the
probability that we encounter case (b) is zero\footnote{Because of the discreteness of random number generators in practice,
as well as the possibility of using a nonuniform sampler for performance,
the probability may not be exactly zero. Nevertheless, randomly resampling
a point or perturbing our initial point if we see case (b) is sufficient.}.

If there are nine tight constraints at our solution, the optimizer
$x^{*}(\theta)$ at this point is the solution to $G_{t}x=w_{t}+F_{t}\theta$,
where $G_{t}$, $w_{t}$, and $F_{t}$ correspond to $G$, $w$, and
$F$ restricted to the tight constraints. The polytope (in parameter
space) within which $x^{*}(\theta)$ is the optimizer is given by
$G_{s}x^{*}(\theta)\le w_{s}+F_{s}\theta$, where $G_{s}$, $w_{s}$,
and $F_{s}$ are the restrictions of $G$, $w$, and $F$ to slack
constraints. Similarly, if there are more than nine tight constraints,
we set the tight constraints to equality and solve to get the optimizer;
the polytope is defined as the region where the slack constraints
remain slack at the optimizer (see Figure \ref{fig:mplp} for details).

Once we have obtained this polytope, we partition the remainder of
our initial polytope into convex polytopes and recursively partition
each one using this algorithm. We terminate when the remainder is
either empty or is of lower dimension than parameter space.

\begin{figure}
\noindent\fbox{\begin{minipage}[t]{1\columnwidth - 2\fboxsep - 2\fboxrule}%
\begin{algorithm}[H]
\DontPrintSemicolon
\KwData{An initial region $R = A\theta \le b$ to explore}
\KwResult{
A set $\{(R_{i},\hat{x}^{i}(\theta))\}$, where $R_i$ form a partition of $R$ and $\hat{x}^{i}(\theta))$ are logs (base $M$) of the optimal tile sizes for $\theta\in R_{i}$
}

\If{$R$ is lower dimension or empty}{
	\KwRet{$\emptyset$}
}

Randomly sample element $\theta_0 \in R$

$x_0^{*} \leftarrow$ optimizer for $\min c^{T}x$ s.t. $Gx\le w+F\theta_{0}$ (solve using simplex)

$A(\theta_0) \leftarrow$ indices of zeros of $Gx_{0}^{*}-F\theta$

$(G_{t},w_{t},F_{t})\leftarrow$ rows $A(\theta_0)$ of $(G,w,F)$

$(G_{s},w_{s},F_{s})\leftarrow$ rows $\{1,...,\vert w\vert\}\backslash A(\theta_0)$ of $(G,w,F)$

\eIf{$\vert A\vert=9$}{
	$\hat{x}^{1}(\theta)\leftarrow G_{t}^{-1}F_{t}\theta+G_{t}^{-1}w_{t}$
}{
	Row-reduce the linear system $\left[\begin{array}{c|c}G_{t} & -F_{t}\end{array}\right]\left[\begin{array}{c}x^{*}(\theta)\\\hline \theta\end{array}\right]=w_t$ to $\left[\begin{array}{c|c}U & P\\\hline 0 & D\end{array}\right]\left[\begin{array}{c}x^{*}(\theta)\\\hline \theta\end{array}\right]=\left[\begin{array}{c}q\\\hline r\end{array}\right]$

\If{$D,r \ne 0$}{
	\tcp{This occurs w.p. 0}
	Resample $\theta_0$ and restart.
}

$\hat{x}^{1}(\theta)\leftarrow-U^{-1}p+U^{-1}q$

}

$R_{1}\leftarrow\{\theta:G_{s}x^{*}(\theta)\le w_{s}+F_{s}\theta\}$

$S_i \leftarrow$ polytope consisting of points that violates the $i$th of $R_{1}$ and satisfies constraints $1$ through $i-1$

Recursively partition each $S_i$ to get set of regions, optimizers $T_i$.

\KwRet{$\{(R_{1}, \hat{x}^{1}(\theta))\} \cup T_1 \cup T_2...$}

\end{algorithm}%
\end{minipage}}

\caption{Algorithm for partitioning parameter space\label{fig:mplp}}
\end{figure}

\subsection{Verifying optimality of the tiling}

In order to verify the optimality of the tiling and the attainability
of the lower bound, we must ensure that the communication cost attained
by the tiling equals the maximum of the five lower bounds for \emph{every}
possible element.

Define the function $\mathscr{C}_{p}(\theta)$ as the log of the communication
cost at point $\theta$ using the tiling provided by the algorithm
above for partition $p$, that is:
\begin{eqnarray}
\mathscr{C}_{p}(\theta) & \coloneqq & \log_{M}B+\log_{M}C+\log_{M}K+\log_{M}W+\log_{M}H+\log_{M}R+\log_{M}S\label{eq:ccost_tile}\\
 &  & +1-\left(\hat{x}_{b}^{p}(\theta)+\hat{x}_{c}^{p}(\theta)+\hat{x}_{k}^{p}(\theta)+\hat{x}_{w}^{p}(\theta)+\hat{x}_{h}^{p}(\theta)+\hat{x}_{r'}^{p}(\theta)+\hat{x}_{r''}^{p}(\theta)+\hat{x}_{s'}(\theta)+\hat{x}_{s''}^{p}(\theta)\right)\nonumber 
\end{eqnarray}
where $\hat{x}_{b}^{p},...,\hat{x}_{s''}^{p}$ represent the closed-form
optimizers in region $p$ as a function of $\theta$, and let $\mathscr{L}_{i}(\theta)$
be the $i$th lower bound from \eqref{eqn_LB1} as a function of $\theta$.

It suffices to show that for every partition $p$ and parameter set
$\theta\in p$, $\mathscr{C}_{p}(\theta)$ is equal to the maximum
of the five lower bounds; in other words, that the quantity
\begin{eqnarray*}
\max_{\theta}\mathscr{C}_{p}(\theta)-\mathscr{L}_{i}(\theta) & s.t.\\
\theta & \in & p\\
\mathscr{L}_{i}(\theta) & \ge & \mathscr{L}_{j}(\theta)\ \ \ \ \ \ \ \forall j\ne i
\end{eqnarray*}
is zero for all $i\in\{1,...,5\}$ and for all partitions $p$. This
is not precisely the case - in the cases where all the inputs and
outputs fit inside fast memory, the result may be nonzero since the
computation in Equation \ref{eq:ccost_tile} assumes that $M$ words
are transmitted in each round (which is obviously more than the number
of words transmitted if everything fits in cache); as a result, we
exclude regions where the communication lower bound is less than $M$.
The LP solutions can easily be checked to be correct using an LP solver;
for our code and results, see \url{https://people.eecs.berkeley.edu/~dinh/papers/DD18/partitioning.nb}.

As a sanity-check, we can also ensure that no tiling generated by
our LP breaks any of our lower bounds, i.e. 
\begin{eqnarray*}
\min_{\theta}\mathscr{C}_{p}(\theta)-\mathscr{L}_{i}(\theta) & s.t.\\
\theta & \in & p
\end{eqnarray*}
should be nonnegative for all $p$, $i$. This is confirmed by our
code as well.

\subsection{Observations}

\begin{figure}[t]
\noindent\fbox{\begin{minipage}[t]{1\columnwidth - 2\fboxsep - 2\fboxrule}%
\includegraphics[width=0.95\textwidth]{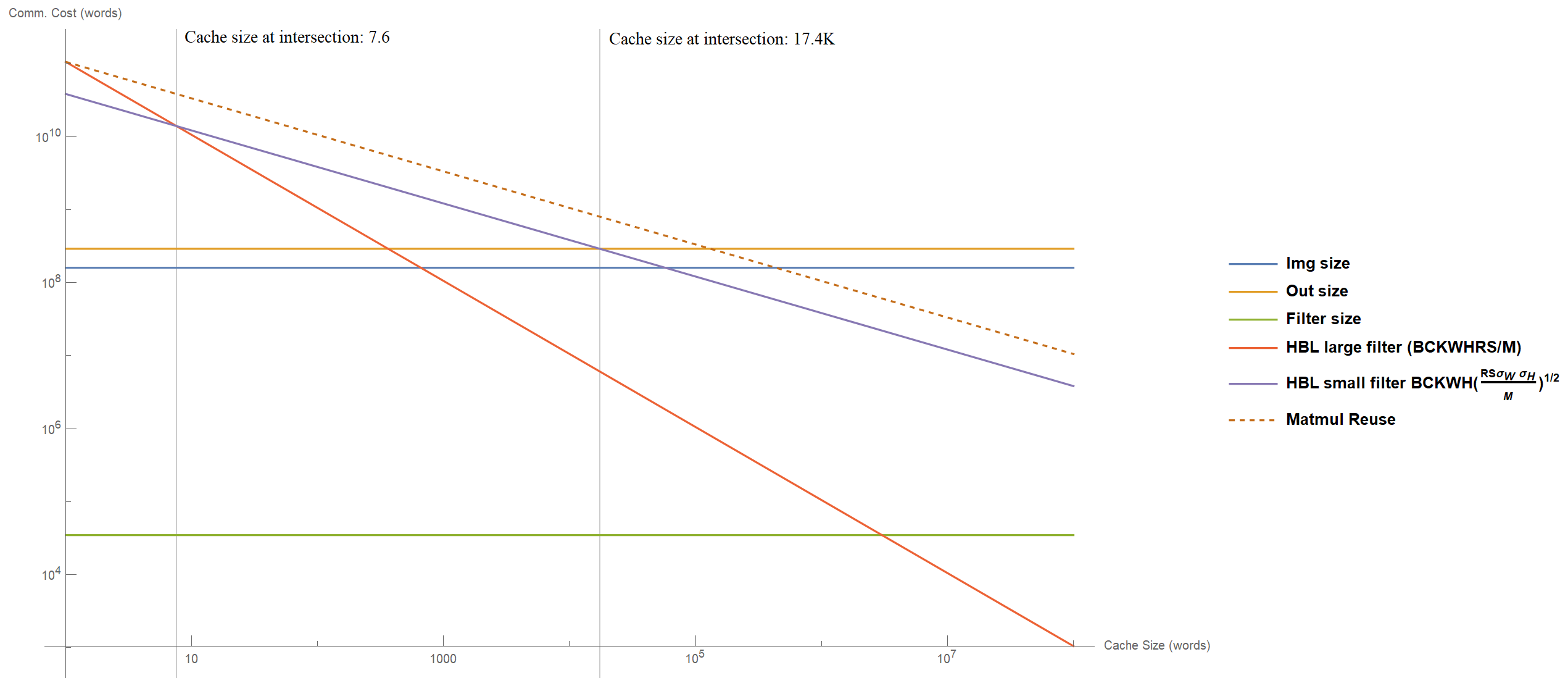}

\caption{Log-log plot of communication cost vs. memory size (both in words)
for AlexNet \cite{KSH17} with $1000$ batches, $B=1000,\ C=3,\ K=96,\ W=H=55,\ R=S=11,\ \sigma_{W}=\sigma_{H}=4$,
compared to a version with the same reuse factor as communication-avoiding
matrix multiply (dashed line), which requires up to $11/4$ times
more communication. Note that that the bound $BCKWH(\frac{RS\sigma_{W}\sigma_{H}}{M})^{1/2}$
is tight up to 17K words, well above the size of the L1 cache (and
often the L2 cache as well) of modern processors. The tiling given
by our algorithm from Section \ref{subsec:opttile} for a 1024-word
(32K or 64K, depending on the size of a number on the architecture)
cache is given by: $b_{b}=b_{k}=\sqrt{M\sigma_{H}\sigma_{W}/(RS)}\approx12$,
$b_{c}=C=3$, $b_{h}=S/\sigma_{H}\approx3$, $b_{w}=R/\sigma_{W}\approx3$,
$b_{s'}=S/\sigma_{H}\approx3$, $b_{s''}=\frac{1}{C}\sqrt{\frac{M\sigma_{H}\sigma_{W}}{RS}}\approx4$,
$b_{r'}=R/\sigma_{W}\approx3$, $b_{r''}=1$. \label{fig:forrest-plot}}
\end{minipage}}
\end{figure}

A plot of our communication costs for one set of CNN parameters for
a real world neural net is shown in Figure \eqref{fig:forrest-plot}.
Notice that different communication bounds (and different tilings)
apply depending on the size of the fast memory being optimized for.
As a result, a more sophisticated analysis would be required to optimize
communication for multi-level memory model with different layers of
cache; we leave this to future work.

It can also be observed in this example that when the memory size
is larger than the size of the filter, the optimal tiling requires
asymptotically no more memory than the size of the output array. In
fact, an examination of the solutions for the LP \eqref{eq:mplp}
shows that in cases where any of $Image$, $Output$, and $Filter$
fits entirely in cache, there always exists a tiling that will ensure
a communication cost no higher than the max of their sizes.

%% file: Pooling.tex
\section{Application to Pooling}

\label{sec_pooling}

Our techniques can also be extended to programs with similar loop
structures, such as pooling:

\begin{align*}
 & {\rm for}\,\{b,c,k,w,h,r,s\}=0:\{B,C,K,W,H,R,S\}-1\\
 & \ \ \ \ \;\;Out(k,h,w,b)\oplus=Image(r+\sigma_{w}w,s+\sigma_{h}h,c,b)
\end{align*}
where $a\oplus b$ can either mean $a+b/(rs)$ (``average-pooling'')
or $\max(a,b)$ (``max-pooling''). The following communication lower
bound holds for pooling:
\[
W_{pool}=\Omega(\max(BKWH,\sigma_{W}\sigma_{H}BCWH,BCKWHRS/M)\ .
\]
The first two terms correspond to the size of $Out$ and $Image$
respectively, while the third can be found using the approach from
Section \ref{subsec_LB2}. In fact, since the only difference between
pooling and the convolution \eqref{eqn_CNN} is the absence of $Filter$,
it suffices to remove $\phi_{3}$ from our calculations. In particular,
the sets of subgroups we consider
\begin{eqnarray*}
C'_{1} & = & \{(0,0,k,0,0,0,0)\}=\{C_{1,1}\}\\
C'_{2} & = & \{(0,0,0,0,0,0,s),(0,0,0,0,h,0,-\sigma_{h}h)\}=\{C_{2,2},C_{2,3}\}\\
C'_{3} & = & \{(0,0,0,0,0,r,0),(0,0,0,w,0,-\sigma_{w}w,0)\}=\{C_{3,2},C_{3,3}\}\\
C'_{4} & = & \{(0,c,0,0,0,0,0)\}=\{C_{4,1}\}
\end{eqnarray*}
Solving the resulting LP gives us the third lower bound, as desired.

We can also follow the approach from Section \ref{sec_UpperBound}
to verify that this lower bound is tight and always attainable with
the tiling given by the solution to the LP \eqref{eq:mplp} with the
\emph{filter} constraint ($l_{c}+l_{k}+l_{r'}+l_{r''}+l_{s'}+l_{s''}\le1$)
removed.

%% file: Conclusions.tex
\section{Conclusions and Future Work}

\label{sec_Conclusions}

We have found an asymptotic lower bound on communication required
to evaluate a convolutional neural net, and provided an optimal reordering
of the nested loops to attain the lower bound in every case by taking
advantage of significantly more ($2.75$ times in the example from
Figure \ref{fig:forrest-plot}) data reuse than is possible with matrix
multiply or many other dense linear algebra operations. We describe
a few avenues for future research below:

\textbf{Parallel Algorithm:} Suppose we have $p$ processors which
do not communicate with each other except through a single shared
memory of size M, and we wish to divide work between them while maintaining
communication efficiency. For simplicity, let us only consider dividing
work by evenly partitioning the output dimensions (i.e. assigning
each processor a subset of $b$, $k$, $h$, and $w$) so as to avoid
possible race conditions caused by assigning the same coordinate in
the output to multiple processors. Suppose we assign each processor
a single block of size $B\lyxmathsym{\textquoteright}=B/a_{B}$, $K\lyxmathsym{\textquoteright}=K/a_{K}$,
$H\lyxmathsym{\textquoteright}=H/a_{H}$, and $W\lyxmathsym{\textquoteright}=W/a_{W}$,
such that $a_{B}a_{K}a_{H}a_{W}\le p$.

Since each of the processors handles a subproblem of the same size,
the communication lower bound (total memory traffic between the shared
memory and all the processors) under these assumptions is $p$ multiplied
by the per-processor expression \eqref{eqn_LB1}, with $B$, $K$,
$H$, $W$ replaced by $B\lyxmathsym{\textquoteright}$, $K\lyxmathsym{\textquoteright}$,
$W\lyxmathsym{\textquoteright}$, $H\lyxmathsym{\textquoteright}$.
The first (output size), fourth, and fifth terms in this expression
remain the same; the second and third (output and filter, respectively)
increase by a factor of $a_{K}$ and $a_{B}a_{W}a_{H}$ respectively;
once these parameters are fixed, the communication lower bound is
\[
\max(BKWH,a_{K}\sigma_{W}\sigma_{H}BCWH,a_{B}a_{W}a_{H}CKRS,BCKWHRS/M,BCKWH(RS\sigma_{W}\sigma_{H}/M)^{1/2})\ .
\]
Since the bound is for fixed values of $a$, we should choose those
values in order to minimize the above quantity.

Optimizing the tiling (under these assumptions) can be done by solving
a modified version of LP \eqref{eq:mplp}, with the $B\lyxmathsym{\textquoteright}$,
$K\lyxmathsym{\textquoteright}$, $W\lyxmathsym{\textquoteright}$,
$H\lyxmathsym{\textquoteright}$ replacing $B$, $K$, $H$, $W$
and with additional constraints encoding $a_{B}a_{K}a_{H}a_{W}\le p$,
$pa_{B}\le B$, etc.

A more sophisticated analysis (examining tiling schemes that may cause
race conditions and distributed models where processors can directly
communicate with each other) is left to future work.

\textbf{Implementation and benchmarking:} The tilings we generate
are asymptotically optimal in terms of communication. However, many
factors, such as cache locality and processor architecture (tile sizes
that are multiples of a processor's vector width are likely to be
more efficient), can provide significant constant-factor changes to
the real-world performance (both in terms of time and energy). Since
our algorithm is simply a rearrangement of the same arithmetic operations
performed during a CNN, we believe our algorithms can provide a significant
advantage over current implementations of convolutions that do no
tiling (e.g. that used in Torch\footnote{\url{https://github.com/torch/torch7/blob/master/lib/TH/generic/THTensorConv.c} })
or those that tile only based on machine parameters without taking
into account the dimensions of the problem (e.g. CuTorch\footnote{\url{https://github.com/torch/cutorch/blob/master/lib/THC/THCTensorConv.cu} });
we leave the validation of this intuition through an implementation
and benchmark to future work.

\textbf{Generalizing to arbitrary nested loops: }Our lower bound in
the ``small filter'' case rests on a problem-specific lifting of
the HBL LP to a higher dimension; similarly, our approach for finding
a closed form for optimal tilings (which is necessary for checking
if a lower bound is always attainable) relies on the creation of a
linear program tailored to this specific problem. Generalizing this
to arbitrary nested loops would move us closer to being able to automatically
optimize arbitrary loop nests for communication, e.g. at a compiler
level.

%% file: UpperBound_B_manual.tex
\section*{Appendix A: Manual Exploration of Tilings}

\label{sec_UpperBound_manual}

In this section we show, by hand, that the lower bounds in Theorem~\ref{thm_LB}
are attainable in the case where $\sigma_{w}=\sigma_{h}=1$ . The
advantage of this approach appears is that it seems to be possible
to attain far more compact representations of the optimal tiling function
than with an automated exploration, which produced well over twice
as many cases (for the one-stride case) as the hand exploration in
this section.

There are a number of cases, depending on which of the five terms
in Theorem~\ref{thm_LB} is largest, and other inequalities. To simplify
the presentation, we present the overall result as a decision tree,
where each leaf of the tree represents a disjoint subset of the possible
lower bounds and optimal algorithms. After stating the result, we
will give some intuition for the decision tree, as arising from a
linear program, and then prove the theorem by providing an algorithm
and communication cost analysis, with a separate lemma for each leaf.

\begin{theorem}\label{thm_UB} The following cases describe which
of the lower bound expressions in Theorem~\ref{thm_LB} are attainable.
The abbreviation ALB stands for Attainable Lower Bound: \begin{tabbing}
blah \= blah \= blah \= blah \= blah \= blah \= blah \kill
\> if $\min(CHWB,KCRS,KHWB)\leq M$ \\
 \> \> Case 1: $ALB=O(\max(CHWB,KCRS,KHWB))$ \\
 \> else (Case 2) \\
 \> \> if $RS\geq M$ \\
 \> \> \> Case 2.1: $ALB=O(KCHWBRS/M)$ \\
 \> \> else (Case 2.2) \\
 \> \> \> if $MRS\geq(BHW)^{2}$ \\
 \> \> \> \> Case 2.2.1: $ALB=O(KCRS)$ \\
 \> \> \> else (Case 2.2.2) \\
 \> \> \> \> if $\min(C,K)\geq(M/(RS))^{1/2}$ \\
 \> \> \> \> \> Case 2.2.2.1: $ALB=O(KCHWB(RS/M)^{1/2})$ \\
 \> \> \> \> else \\
 \> \> \> \> \> Case 2.2.2.2: $ALB=O(\max(KHWB,CHWB))$ \end{tabbing}
\end{theorem}

Here is some intuition for why the cases in Theorem~\ref{thm_UB}
arise, and some notation we will use later in the proof. Suppose we
tile the $b$ loop with block size $bB$, the $k$ loop with $bK$,
the $h$ loop with $bH$, and so on. Then the algorithm becomes (recall
that in this case we are assuming $\sigma_{w}=\sigma_{h}=1$):
\begin{align*}
 & \text{Algorithm }BlockCNN(b|bB,c|bC,k|bK,w|bW,h|bH,r|bR,s|bS,):\\
 & \ \ {\rm for}\,\{b,c,k,w,h,r,s\}_{1}=0:b_{\{b,c,k,w,h\}}:\{B,C,K,W,H,R,S\}-b_{\{b,c,k,w,h,r,s\}}\\
 & \ \ \ \ {\rm for}\,\{b,c,k,w,h,r,s\}_{2}=0:b_{\{b,c,k,w,h,r,s\}}-1\\
 & \;\;\;\;\ \ \{b,c,k,w,h,r,s\}=\{b,c,k,w,h,r,s\}_{1}+\{b,c,k,w,h,r,s\}_{2}\\
 & \;\;\;\;\ \ Out(k,h,w,b)+=Image(r''+\sigma_{w}(r'+w),\ s''+\sigma_{h}(s'+h),c,b)\\
 & \;\;\;\;\ \ \phantom{Out(k,h,w,b)+=}\times Filter(k,r'+r'',\ s'+s'',s,c)
\end{align*}

The argument list in the name is meant to indicate both the order
of the nested loops, via the notation $b|$, $k|$, $h|$, etc., and
the block sizes. A natural way to try to optimize this code is to
pick the block sizes so that all the data accessed in the 7 innermost
loops fits in fast memory of size $M$, and to maximize the number
of loop iterations that can be performed by these loops, namely $bB\cdot bK\cdot bH\cdot bW\cdot bR\cdot bS\cdot bC$.
It is easy to see that the submatrix of $Out$ accessed by these loops
is of size $bK\cdot bH\cdot bW\cdot bB$, the submatrix of $Filter$
accessed is of size $bK\cdot bR\cdot bS\cdot bC$, and the submatrix
of $Image$ accessed is of size $bC\cdot(bS+bH)\cdot(bR+bW)\cdot bB$.
Since $R\leq W$, we will assume $bR\leq bW$, and similarly $bS\leq bH$,
which means the submatrix of $C$ accessed is of size at most $4\cdot bC\cdot bH\cdot bW\cdot bB$;
we will ignore the constant 4 for simplicity, since it does not change
our Big-O analysis below. This yields the following optimization problem:
\begin{tabbing} blah \= blah \= blah \= blah \= blah \= blah
\= blah \kill \> maximize $G=bB\cdot bK\cdot bH\cdot bW\cdot bR\cdot bS\cdot bC$
\\
 \> subject to the constraints \\
 \> \> $bK\cdot bH\cdot bW\cdot bB\leq M$ \\
 \> \> $bC\cdot bH\cdot bW\cdot bB\leq M$ \\
 \> \> $bK\cdot bR\cdot bS\cdot bC\leq M$ \\
 \> \> $bR\leq bW$, $bS\leq bH$ \\
 \> \> $bB\leq B$, $bK\leq K$, $bH\leq H$, $bW\leq W$, $bR\leq R$,
$bS\leq S$, $bC\leq C$ \\
 \> \> $bB\geq1$, $bK\geq1$, $bH\geq1$, $bW\geq1$, $bR\geq1$,
$bS\geq1$, $bC\geq1$ \end{tabbing} We call a tiling {\em admissible}
if it satisfies the constraints above. (The reader may wonder why
the first 3 constraints do not use the upper bound $M/3$ instead
of $M$, to be sure all three submatrices fit in fast memory simultaneously;
this would only change the value of $G$ by a constant factor, which
would not change the following Big-O analysis.) Then assuming the
3 submatrices are each read or written just once in the innermost
7 loops, the total number of reads and writes is $O(BKHWCRSM/G)$,
since $BKHWCRS$ is the total number of loop iterations, $BKHWCRS/G$
is the number of iterations of the outermost 7 loops, and there are
$O(M)$ reads/writes per iteration of the outermost 7 loops. This
last statement about $O(M)$ read/writes per iteration may depend
on the cache replacement policy in the hardware, but again it will
not change the Big-O analysis.

If we now replace each quantity by its logarithm base $M$, so $B$
by $lB=\log_{M}B$, $bB$ by $lbB=\log_{M}bB$ and so on, we get the
following linear program: \begin{tabbing} blah \= blah \= blah
\= blah \= blah \= blah \= blah \kill \> maximize $lG=lbB+lbK+lbH+lbW+lbR+lbS+lbC$
\\
 \> subject to the constraints \\
 \> \> $lbK+lbH+lbW+lbB\leq1$ \\
 \> \> $lbC+lbH+lbW+lbB\leq1$ \\
 \> \> $lbK+lbR+lbS+lbC\leq1$ \\
 \> \> $lbR\leq lbW$, $lbS\leq lbH$ \\
 \> \> $lbB\leq lB$, $lbK\leq lK$, $lbH\leq lH$, $lbW\leq lW$,
$lbR\leq lR$, $lbS\leq lS$, $lbC\leq lC$ \\
 \> \> $lbB\geq0$, $lbK\geq0$, $lbH\geq0$, $lbW\geq0$, $lbR\geq0$,
$lbS\geq0$, $lbC\geq0$ \end{tabbing} Exploring the finite number
of corners of the polytope defined by this linear program lead to
the cases in Theorem~\ref{thm_UB}. While the proof of Theorem~\ref{thm_UB}
requires this exploration by hand, and confirming that the lower bound
of Theorem~\ref{thm_LB} is attained, in practice the linear program
could be used to determine the optimal block sizes. We note that this
linear program has 7 variables and 12 constraints (besides nonnegativity),
so there are as many as $\binom{12}{7}=792$ corners in the feasible
polytope to explore; fortunately only a few turn out to be important.

To keep the proofs short, we will use the notation BlockCNN above,
and the cost expression $O(BKHWCRSM/G)$, to describe the optimal
algorithm in each case. We will use the expression 
\[
LB=\max(KHWB,CHWB,KCRS,KCHWBRS/M,KCHWB(RS/M)^{1/2})
\]
to denote the lower bound from Theorem~\ref{thm_LB}. To capture
(some of) the non-uniqueness of the optimal solutions, we will use
the following two functions to help solve linear programs: Function
$(x,y)=f_{2}(\bar{x},\bar{y},s)$ takes 3 nonnegative arguments satisfying
$\bar{x}+\bar{y}\geq s$, and returns some $0\leq x\leq\bar{x}$ and
$0\leq y\leq\bar{y}$ satisfying $x+y=s$. Function $(x,y,z)=f_{3}(\bar{x},\bar{y},\bar{z},s)$
similarly takes 4 nonnegative arguments satisfying $\bar{x}+\bar{y}+\bar{z}\geq s$,
and returns some $0\leq x\leq\bar{x}$, $0\leq y\leq\bar{y}$ and
$0\leq z\leq\bar{z}$ satisfying $x+y+z=s$.

\begin{lemma}\label{lemma_UB_Case1} \textbf{Upper Bound Case 1:}
Suppose $\min(CHWB,KCRS,KHWB)\leq M$, i.e. at least one of the 3
arrays $Image$, $Filter$ and $Out$ fits in fast memory. Then the
attainable communication lower bound is $O(\max(CHWB,KCRS,KHWB))$.
\end{lemma}

\begin{proof} Case 1 in turn breaks down into a number of subcases
(again, we ignore constant factors): 
\begin{description}
\item [{Case 1.1.:  $CHWB \leq M$, $KCRS \leq M$, $KHWB \leq M$.}] $ $
\linebreak{}
Use BlockCNN($b|B,k|K,h|H,w|W,r|R,s|S,c|C$), i.e. the original unblocked
algorithm,

$LB=\max(CHWB,KHWB,KCRS)$ because $KCRS\leq M$ implies $RS\leq M$
implies 
\[
KCHWBRS/M\leq KCHWB(RS/M)^{1/2}\leq KCHWB(KC)^{-1/2}=(KC)^{1/2}HWB\leq\max(K,C)HWB
\]

\item [{Case 1.2.:  $CHWB \geq M$, $KCRS \leq M$, $KHWB \leq M$.}] $ $
\linebreak{}
Use BlockCNN($b|1,k|K,h|H,w|W,r|R,s|S,c|(M/(HW)$). 
It is straightforward to confirm that this tiling is admissible. Then
$G=KHWRS(M/(HW))=KRSM$, and so the number of read/writes is $O(CHWB)$.
The same inequalities as in Case~1.1 show 
\[
KCHWBRS/M\leq KCHWB(RS/M)^{1/2}\leq\max(K,C)HWB=CHWB
\]
so that $LB=\max(CHWB,KHWB,KCRS)$. 
\item [{Case 1.3.:  $CHWB \leq M$, $KCRS \leq M$, $KHWB \geq M$.}] $ $
\linebreak{}
Swap the roles of $C$ and $K$ in Case 1.2. 
\item [{Case 1.4.:  $CHWB \leq M$, $KCRS \geq M$, $KHWB \leq M$, $KC \leq M$.}] $ $
\linebreak{}
Let $(lbR,lbS)=f_{2}(lR,lS,1-lK-lC)$, and then $bR=M^{lbR}$ and
$bS=M^{lbS}$. Then use BlockCNN($b|B,k|K,h|H,w|W,r|bR,s|bS,c|C$).
Admissibility follows from the definition of $f_{2}()$. Then $G=KHWB(M/KC)C=HWMB$,
and the number of reads/writes is $O(KCRS)$.

To show $LB=\max(CHWB,KHWB,KCRS)$ we first note $CHWB\leq M$ implies
$HWB\leq M$ implies $KCHWBRS/M\leq KCRS$. Multiplying the first
3 inequalities defining Case 1.4 yields $KCRS\cdot M\cdot M\geq M\cdot CHWB\cdot KHWB$,
or $RS\geq(RS/M)^{1/2}HWB$, and thus $KCRS\geq KCHWB(RS/M)^{1/2}$. 
\item [{Case 1.5.:  $CHWB \leq M$, $KCRS \geq M$, $KHWB \leq M$, $KC \geq M$.}] $ $
\linebreak{}
$RS\leq HW\leq M$, so $lR+lS\leq1$. Let $(lbK,lbC)=f_{2}(lK,lC,1-lR-lS)$,
and then $bK=M^{lbK}$ and $bC=M^{lbC}$. Then use BlockCNN($b|B,k|bK,h|H,w|W,r|R,s|S,c|bC$).
Admissibility follows from the definition of $f_{2}()$. Then $G=(M/(RS))HWRSB=HWMB$,
and the number of reads/writes is $O(KCRS)$.

Showing $LB=\max(CHWB,KHWB,KCRS)$ is identical to Case 1.4. 
\item [{Case 1.6: $CHWB \geq M$, $KCRS \leq M$, $KHWB \geq M$.}] $ $
\linebreak{}
$KC\leq M$ so $\max(K,C)\leq M$. 
Let $(lbB,\delta_{h},\delta_{w})=f_{3}(lB,lH-lS,lW-lR,1-\max(lC,lK)-lR-lS)$.
This is well-defined because $lH-lS\geq0$ is equivalent to $H\geq S$,
$lW-lR\geq0$ is equivalent to $W\geq R$, $1-\max(lC,lK)-lR-lS\geq0$
is equivalent to $M\geq\max(KRS,CRS)$, which is implied by $KCRS\leq M$,
and $lB+lH-lS+lW-lR\geq1-\max(lC,lK)-lR-lS$ is equivalent to $\max(KHWB,CHWB)\geq M$.
Now let $lbH=lS+\delta_{h}\leq lH$, $lbW=lR+\delta_{w}\leq lW$,
$bB=M^{lbB}$, $bH=M^{lbH}$ and $bW=M^{lbW}$. Thus $S\leq bH\leq H$
and $R\leq bW\leq W$. Then use BlockCNN($b|bB,k|K,h|bH,w|bW,r|R,s|S,c|C$).
Admissibility follows from $KCRS\leq M$ and $lbB+lbH+lbW=1-\max(lC,lK)$,
so $bB\cdot bH\cdot bW=M/\max(K,C)$, and $\max(K\cdot bH\cdot bW\cdot bB,C\cdot bH\cdot bW\cdot bB)=M$.
Then $G=KC(M/\max(K,C))RS=\min(K,C)MRS$, and the number of reads/writes
is $O(\max(KHWB,CHWB))$.

To show $LB=\max(CHWB,KHWB,KCRS)$ we note \\
$KCHWBRS/M\leq HWB\leq\max(CHWB,KHWB)$ and \\
$KCHWB(RS/M)^{1/2}\leq KCHWB(1/(KC))^{1/2}=(KC)^{1/2}HWB\leq\max(K,C)HWB$. 
\item [{Case 1.7: $CHWB \geq M$, $KCRS \geq M$, $KHWB \leq M$, $K \geq HWB$.}] $ $
\linebreak{}
Use BlockCNN($b|B,k|K,h|H,w|W,r|1,s|1,c|(M/K)$). 
Admissibility follows from $(M/K)HWB\leq M$. Then $G=MHWB$ and the
number of reads/writes is $O(KCRS)$.

Note that $KCRS\geq HWBCRS\geq CHWB\geq KHWB$ and $KCHWBRS/M\leq KCRS$.
$KHWB\leq M$ and $K\geq HWB$ together imply $HWB\leq M^{1/2}\leq(MRS)^{1/2}$,
and hence $KCHWB(RS/M)^{1/2}\leq KCRS$. So altogether $LB=\max(KCRS,CHWB,KHWB)$. 
\item [{Case 1.8: $CHWB \geq M$, $KCRS \geq M$, $KHWB \leq M$, $K \leq HWB$.}] $ $
\linebreak{}
Let $(lbR,lbS)=f_{2}(lR,lS,min(lR+lS,lH+lW+lB-lK))$, and then $bR=M^{lbR}$
and $bS=M^{lbS}$. Note that $lH+lW+lB-lK\geq0$ because $K\leq HWB$.

Use BlockCNN($b|B,k|K,h|H,w|W,r|bR,s|bS,c|(M/(HWB))$. Admissibility
follows since \\
$K\cdot(M/(HWB))\cdot bR\cdot bS\leq K\cdot(M/HWB)\cdot(HWB/K)=M$.
Then \\
$G=BKHW(M/HWB)\min(RS,HWB/K)=\min(KMRS,MHWB)$ and the number of reads/writes
is $O(KCRSHWBM/\min(KMRS,MHWB))=O(\max(CHWB,KCRS))$.

$KHWB\leq M$ implies $HWB\leq M$ implies $KCRSHWB/M\leq KCRS$.
If $KRS\leq HWB$, then $\frac{K(RS)^{1/2}}{M^{1/2}}\leq\frac{K(HWB/K)^{1/2}}{M^{1/2}}=\frac{(KHWB)^{1/2}}{M^{1/2}}\leq1$,
implying $CHWB\geq KCHWB(RS/M)^{1/2}$. Alternatively, if $KRS\geq HWB$,
then $M\geq KHWB\geq(HWB)^{2}/RS$, so $RS(HWB)^{2}/M\leq(RS)^{2}$,
and $KCRS\geq KCHWB(RS/M)^{1/2}$. So altogether $LB=\max(KCRS,CHWB,KHWB)$. 
\item [{Case 1.9: $CHWB \leq M$, $KCRS \geq M$, $KHWB \geq M$, $C \geq HWB$.}] $ $
\linebreak{}
Swap the roles of $C$ and $K$ in Case 1.7. 
\item [{Case 1.10: $CHWB \leq M$, $KCRS \geq M$, $KHWB \geq M$, $C \leq HWB$.}] $ $
\linebreak{}
Swap the roles of $C$ and $K$ in Case 1.8. 
\end{description}
\end{proof}

\begin{lemma}\label{lemma_UB_Case2.1} \textbf{Upper Bound Case 2.1:}
Suppose $RS\geq M$. Then the attainable communication lower bound
is $O(KCHWRSB/M)$. \end{lemma}

\begin{proof} Note that $HW\geq RS\geq M$ implies $\min(CHWB,KHWB,KCRS)\geq M$.
Let $(lbR,lbS)=f_{2}(lR,lS,1)$, $lbW=lbR$ and $lbH=lbS$, and then
$bR=M^{lbR}$, $bS=M^{lbS}$, $bW=bR$ and $bH=bS$. Use BlockCNN($b|1,k|1,c|1,h|bH,w|bW,r|bR,s|bS$).
Admissibility follows from $bR\cdot bS=bH\cdot bW=M$. Then $G=M^{2}$,
so the number of reads/writes is $O(KCHWRSMB/G)=O(KCHWRSB/M)$.

$HW\geq RS\geq M$ implies $KCHWRSB/M\geq\max(CHWB,KHWB,KCRS)$. $RS\geq M$
also implies $KCHWBRS/M\geq KCHWB(RS/M)^{1/2}$. So $LB=KCHWRSB/M$.
\end{proof}

\begin{lemma}\label{lemma_UB_Case2.2.1} \textbf{Upper Bound Case
2.2.1:} Suppose $\min(CHWB,KCRS,KHWB)\geq M$, \\
$RS\leq M$ and $MRS\geq(HWB)^{2}$. Then the attainable communication
lower bound is $O(KCRS)$. \end{lemma}

\begin{proof} Note that $MHW\geq MRS\geq(HWB)^{2}$, so $M\geq B^{2}HW$.
Let $bC=M/(HWB)\geq1$, and so $bC\leq CHWB/(HWB)=C$. Also $KHWB\geq M\geq(HWB)^{2}/(RS)$
so $K\geq HWB/(RS)$. Let $bK=HWB/(RS)\geq1$, and so $bK\leq K$.
Use BlockCNN($b|B,k|bK,h|H,w|W,r|R,s|S,c|bC$). Admissibility follows
from $bC\cdot HWB=M$, $bK\cdot HWB=(HWB)^{2}/(RS)\leq M$, and $bC\cdot bK\cdot RS=M/(HWB)\cdot HWB/(RS)\cdot RS=M$.
Then $G=bK\cdot bC\cdot HWRSB=HWMB$, so the number of reads/writes
is $O(KCHWRSBM/(HWMB))=O(KCRS)$.

$(HWB)^{2}\leq MRS$ implies $HWB(RS/M)^{1/2}\leq RS$ implies $KCHWB(RS/M)^{1/2}\leq KCRS$.
$RS\leq M$ implies $KCHWBRS/M\leq KCHWB(RS/M)^{1/2}\leq KCRS$. $(HWB)^{2}/RS\leq M\leq KHWB$
implies $HWB\leq KRS$ implies $CHWB\leq KCRS$. Similarly, $(HWB)^{2}/RS\leq M\leq CHWB$
implies $HWB\leq CRS$ implies $KHWB\leq KCRS$. Thus $LB=KCRS$.
\end{proof}

\begin{lemma}\label{lemma_UB_Case2.2.2.1} \textbf{Upper Bound Case
2.2.2.1:} Suppose $\min(CHWB,KCRS,KHWB)\geq M$, \\
$RS\leq M$, $MRS\leq(HWB)^{2}$ and $\min(C,K)\geq(M/(RS))^{1/2}$.
Then the attainable communication lower bound is $O(KCHWB(RS/M)^{1/2})$.
\end{lemma}

\begin{proof} Let $bK=bC=(M/(RS))^{1/2}\leq\min(C,K)$, and $bR=R$
and $bS=S$. \linebreak{}
Let $(lbB,\delta_{h},\delta_{w})=f_{3}(lB,lH-lS,lW-lR,(1-lR-lS)/2)$.
This is well-defined because $lH-lS\geq0$ is equivalent to $H\geq S$,
$lW-lR\geq0$ is equivalent to $W\geq R$, $(1-lR-lS)/2\geq0$ is
equivalent to $(M/RS)^{1/2}\geq1$ or $M\geq RS$, and $lB+lH-lS+lW-lR\geq(1-lR-lS)/2$
is equivalent to $HWB/(RS)\geq(M/RS)^{1/2}$, or $(HWB)^{2}\geq MRS$.
Now let $lbH=lS+\delta_{h}\leq lH$, $lbW=lR+\delta_{w}\leq lW$,
$bH=M^{lbH}$, $bW=M^{lbW}$ and $bB=M^{lbB}$. Then $bS\leq bH\leq H$,
$bR\leq bW\leq W$, and $bH\cdot bW\cdot bB=(MRS)^{1/2}$. Use BlockCNN($b|bB,k|bK,h|bH,w|bW,r|bR,s|bS,c|bC$).
Admissibility follows from $bK\cdot bH\cdot bW\cdot cB=bC\cdot bH\cdot bW\cdot bB=bK\cdot bC\cdot bR\cdot bS=M$.
Then $G=bB\cdot bK\cdot bC\cdot bH\cdot bW\cdot RS=M^{3/2}(RS)^{1/2}$,
and the number of reads/writes is $O(KCHWRSBM/(M^{3/2}(RS)^{1/2}))=O(KCHWB(RS/M)^{1/2})$.

$K\geq(M/(RS))^{1/2}$ implies $KCHWB(RS/M)^{1/2}\geq CHWB$. $C\geq(M/(RS))^{1/2}$
implies \linebreak{}
$KCHWB(RS/M)^{1/2}\geq KHWB$. $HWB\geq(MRS)^{1/2}$ implies $KCHWB(RS/M)^{1/2}\geq KCRS$.
$RS\leq M$ implies $KCHWB(RS/M)^{1/2}\geq KCHWBRS/M$. Thus $LB=KCHWB(RS/M)^{1/2}$.
\end{proof}

\begin{lemma}\label{lemma_UB_Case2.2.2.2} \textbf{Upper Bound Case
2.2.2.2:} Suppose $\min(CHWB,KCRS,KHWB)\geq M$, \\
$RS\leq M$, $MRS\leq(HWB)^{2}$ and $\min(C,K)\leq(M/(RS))^{1/2}$.
Then the attainable communication lower bound is $O(\max(KHWB,CHWB))$.
\end{lemma}

\begin{proof} Suppose w.l.o.g that $C\leq K$, so $C\leq(M/(RS))^{1/2}$
and $KCRS\geq M$ implies \linebreak{}
$K\geq M/(CRS)\geq(M/(RS))^{1/2}$. Let $bC=C$, $bK=C$, $bR=R$
and $bS=S$. \linebreak{}
Let $(lbB,\delta_{h},\delta_{w})=f_{3}(lB,lH-lS,lW-lR,1-lC-lR-lS)$.
This is well-defined because $lH-lS\geq0$ is equivalent to $H\geq S$,
$lW-lR\geq0$ is equivalent to $W\geq R$, $1-lC-lR-lS\geq0$ is equivalent
to $M/(CRS)\geq1$, which is implied by $C\leq(M/RS)^{1/2}\leq M/RS$,
and $lB+lH-lS+lW-lR\geq1-lC-lR-lS$ being equivalent to $lC+lH+lW+lB\geq1$
or $CHWB\geq M$. Now let $lbH=lS+\delta_{h}\leq lH$, $lbW=lR+\delta_{w}\leq lW$,
$bH=M^{lbH}$, $bW=M^{lbW}$ and $bB=M^{lbB}$. Then $bS\leq bH\leq H$,
$bR\leq bW\leq W$, and $bH\cdot bW\cdot bB=M/C$. Use BlockCNN($b|bB,k|bK,h|bH,w|bW,r|bR,s|bS,c|bC$).
Admissibility follows from $bK\cdot bH\cdot bW\cdot bB=bC\cdot bH\cdot bW\cdot bB=C\cdot M/C=M$,
and $bK\cdot bC\cdot bR\cdot bS=C^{2}RS\leq M$. Then $G=bB\cdot bK\cdot bC\cdot bH\cdot bW\cdot bR\cdot bS=C\cdot C\cdot M/C\cdot R\cdot S=CRSM$,
and the number of reads/writes is $O(KCHWRSBM/(CRSM))=O(KWHB)$.

$KHWB\geq CHWB$ by assumption. $HWB\geq(MRS)^{1/2}\geq CRS$ implies
$KHWB\geq KCRS$. \linebreak{}
$C\leq(M/(RS))^{1/2}\leq M/(RS)$ implies $KHWB\geq KCHWB(RS/M)^{1/2}\geq KCHWBRS/M$.
Thus $LB=KHWB$. \end{proof}

This completes the proof of Theorem~\ref{thm_UB}.